\documentclass[lettersize,journal]{IEEEtran}
\usepackage{amsmath,amsfonts,amssymb,amsthm}
\newtheorem*{theorem*}{Theorem}
\usepackage{algorithmic}
\usepackage[algoruled, vlined, linesnumbered]{algorithm2e}
\usepackage{array}
\usepackage[caption=false,font=normalsize,labelfont=sf,textfont=sf]{subfig}
\usepackage{textcomp}
\usepackage{stfloats}
\usepackage{url}
\usepackage{verbatim}
\usepackage{graphicx}
\usepackage{cite}
\usepackage[dvipsnames]{xcolor}
\usepackage{hyperref}
\usepackage{multirow}
\usepackage{dsfont}
\usepackage{microtype}
\usepackage{amsmath,amssymb}

\newtheorem{theorem}{Theorem}

\newtheorem{lemma}{Lemma}

\usepackage{stmaryrd}
\newtheorem{assumption}{Assumption}

\DeclareSymbolFont{sansops}{OT1}{\sfdefault}{m}{n}
\SetSymbolFont{sansops}{bold}{OT1}{\sfdefault}{b}{n}

\makeatletter
\renewcommand\operator@font{\mathgroup\symsansops}
\makeatother

\DeclareSymbolFont{sfoperators}{OT1}{cmss}{m}{n}
\DeclareSymbolFontAlphabet{\mathsf}{sfoperators}

\makeatletter
\def\operator@font{\mathgroup\symsfoperators}
\makeatother

\newcommand{\R}{\mathbb{R}}
\newcommand{\C}{\mathbb{C}}

\renewcommand{\Re}{\operatorname{Re}}%\left\{#1\right\}}
\renewcommand{\Im}{\operatorname{Im}}%\left\{#1\right\}}
\newcommand{\conj}[1]{\mkern 1.5mu\underline{\mkern-1.5mu#1\mkern-1.5mu}\mkern 1.5mu}

\newcommand{\vct}[1]{\boldsymbol{#1}}
\newcommand{\mtx}[1]{\boldsymbol{#1}}
\newcommand{\diag}{\operatorname{diag}}
\newcommand{\<}{\langle}
\renewcommand{\>}{\rangle}
\newcommand{\T}{\top}

\newcommand{\ip}[2]{\left\<#1, #2\right\>}

\newcommand{\p}[1]{\left(#1\right)}

\renewcommand{\c}[1]{\left\{#1\right\}}
\newcommand{\abs}[1]{\left|#1\right|}

\newcommand{\set}[1]{\mathcal{#1}}
\newcommand{\linop}[1]{\mathsf{#1}}    % general linear operator

\newcommand{\va}{\vct{a}}
\newcommand{\vb}{\vct{b}}
\newcommand{\vc}{\vct{c}}

\newcommand{\ve}{\vct{e}}
\newcommand{\vf}{\vct{f}}
\newcommand{\vg}{\vct{g}}

\newcommand{\vk}{\vct{k}}

\newcommand{\vp}{\vct{p}}
\newcommand{\vq}{\vct{q}}

\newcommand{\vs}{\vct{s}}

\newcommand{\vu}{\vct{u}}
\newcommand{\vv}{\vct{v}}
\newcommand{\vw}{\vct{w}}
\newcommand{\vx}{\vct{x}}
\newcommand{\vy}{\vct{y}}
\newcommand{\vz}{\vct{z}}

\newcommand{\vgamma}{\vct{\gamma}}
\newcommand{\vdelta}{\vct{\delta}}

\newcommand{\vkappa}{\vct{\kappa}}

\newcommand{\vphi}{\vct{\phi}}

\newcommand{\vzero}{\vct{0}}
\newcommand{\vone}{\vct{1}}

\newcommand{\mA}{\mtx{A}}
\newcommand{\mB}{\mtx{B}}

\newcommand{\mE}{\mtx{E}}
\newcommand{\mF}{\mtx{F}}
\newcommand{\mG}{\mtx{G}}

\newcommand{\mJ}{\mtx{J}}
\newcommand{\mK}{\mtx{K}}

\newcommand{\mY}{\mtx{Y}}

\newcommand{\mId}{\mathbf{I}}

\newcommand{\loA}{\linop{A}}

\newcommand{\loC}{\linop{C}}

\newcommand{\setE}{\set{E}}

\newcommand{\setG}{\set{G}}

\newcommand{\setN}{\set{N}}

\newcommand{\setP}{\set{P}}

\newcommand{\setS}{\set{S}}

\newcommand{\setV}{\set{V}}

\newcommand{\setZ}{\set{Z}}

\newcommand{\st}{\operatorname{\sf s.t. }}

\newcommand{\ow}{\operatorname{\sf otherwise }}

\usepackage{adjustbox}
\usepackage{booktabs}
\usepackage{arydshln}
\usepackage{orcidlink}
\usepackage{tikz}
\usetikzlibrary{shapes, arrows, quotes, calc}
\usepackage[straightvoltages, americancurrents, smartlabels]{circuitikz}

\tikzset{breaker/.style={generic, %
        bipoles/generic/width=#1, bipoles/generic/height=#1,
    },
    breaker/.default=0.20
}
\usepackage{hyperref}
\hypersetup{
    colorlinks=false,
    linkcolor=blue,
    filecolor=magenta,      
    urlcolor=cyan,
    pdfpagemode=FullScreen,
    }
\newcommand{\uj}{\mathrm{j}}%{\mathbf{j}}
\newcommand{\vvhat}{\vct{\hat{v}}}
\newcommand{\vell}{\vct{\ell}}
\newcommand{\vellhat}{\vct{\hat{\ell}}}

\title{Differentiating Through Power Flow Solutions for Admittance and Topology Control}
\author{Samuel Talkington \orcidlink{0000-0001-5768-8115}, Daniel Turizo \orcidlink{0000-0001-7324-0185}, Sergio A. Dorado-Rojas \orcidlink{0000-0002-5372-1763}, Rahul K. Gupta \orcidlink{0000-0002-9905-6092}, Daniel K. Molzahn \orcidlink{0000-0003-0583-5376}
\thanks{
S. Talkington, D. Turizo, S. A. Dorado-Rojas, and D. K. Molzahn are with the School of Electrical and Computer Engineering, Georgia Institute of Technology, Atlanta, GA, USA. 
Email: \{\href{mailto:talkington@gatech.edu}{talkington}, \href{mailto:djturizo@gatech.edu}{djturizo}, \href{mailto:sadr@gatech.edu}{sadr}, \href{mailto:molzahn@gatech.edu}{molzahn}\}@gatech.edu.
}%
\thanks{
R. K. Gupta is with the School of Electrical Engineering and Computer Science, Washington State University, Pullman, WA, USA. Email: \href{mailto:rahul.k.gupta@wsu.edu}{rahul.k.gupta@wsu.edu}.
}%
\thanks{
This material is based upon work supported by the National Science Foundation Graduate Research Fellowship Program under Grant No. DGE-1650044. Any opinions, findings, and conclusions or recommendations expressed in this material are those of the author(s) and do not necessarily reflect the views of the National Science Foundation.
}
}

\begin{document}

\maketitle

% ============================================================
% ABSTRACT AND KEYWORDS
% ============================================================
\begin{abstract}
    The power flow equations relate bus voltage phasors to power injections via the network admittance matrix. These equations are central to the key operational and protection functions of power systems (e.g., optimal power flow scheduling and control, state estimation, protection, and fault location, among others). As control, optimization, and estimation of network admittance parameters are central to multiple avenues of research in electric power systems, we propose a linearization of power flow solutions obtained by implicitly differentiating them with respect to the network admittance parameters. This is achieved by utilizing the implicit function theorem, in which we show that such a differentiation is guaranteed to exist under mild conditions and is applicable to generic power systems (radial or meshed). The proposed theory is applied to derive sensitivities of complex voltages, line currents, and power flows. The developed theory of linearizing the power flow equations around changes in the complex network admittance parameters has numerous applications. We demonstrate several of these applications, such as predicting the nodal voltages when the network topology changes without solving the power flow equations. We showcase the application for continuous admittance control, which is used to increase the hosting capacity of a given distribution network.
\end{abstract}


% ============================================================
% INTRODUCTION
% ============================================================
\section{Introduction}
\IEEEPARstart{T}{he} power flow equations are widely used for modeling the relationship between the voltage phasors and the power injections in power systems, and are fundamental for key operational and protection functions of power systems. These are used for a number of applications such as state estimation~\cite{schweppe1974static}, scheduling and control based on the optimal power flow (OPF)~\cite{mercier2009optimizing, celli2005multiobjective}, fault location (e.g.,~\cite{saha2010fault}), and protection (e.g.,~\cite{perez2001optimal}), among others. The power flow equations are inherently nonlinear, posing significant computational challenges. This is poignantly true in combinatorial optimization problems, where discrete variables compound with the non-convex AC OPF problem, which can result in an optimization problem that becomes intractable as the size of the network grows. A class of such problems that have attracted considerable interest in the power system literature\textemdash and network science more broadly\textemdash is termed as \emph{topology control} problems. Examples of problems in this class include line switching problems for power shut-offs~\cite{mazi1986corrective}, loss reduction~\cite{bacher1988, fliscounakis2007, fliscounakis2009}, distribution network reconfiguration~\cite{baran-reconfiguration-1989,taylor-convex-reconfiguration-2012,gorka_efficient_radiality_2022}, congestion management via bus splitting~\cite{koglin1982, babaeinejadsarookolaee2022a, babaeinejadsarookolaee_congestion_2023}, etc.

Moreover, a conceptually similar topic\textemdash flexible AC transmission systems (FACTS)\textemdash has been of interest at the transmission scale for a significant period where continuous control of shunt reactors is exercised \cite{peelo1996new, xiao2003available, sahraei2015day}. Existing literature (e.g., \cite{divan2004distributed, johal2007design, kreikebaum2010smart}) suggests the use of series line reactors for active power flow control in power systems. Recent work has developed continuous models \cite{agarwal_continuous_switch_2023,babaeinejadsarookolaee_congestion_2023}. These approaches, belonging to the \emph{continuous admittance} framework, are inclusive of models for transmission line switches, node-breaker substations, and other technologies such as smart wires and unified power flow controllers (UPFCs). This framework has seen application in combinatorial switching problems \cite{agarwal_continuous_switch_2023} and current congestion management \cite{babaeinejadsarookolaee_congestion_2023}. Overall, both the topology control and the continuous admittance problems in power systems are nonlinear and non-convex, often resulting in intractable formulations.

% \paragraph*{Related work}
% \label{sec:related-work}
\textit{Related work:}
% Numerous problems involving the control of admittance parameters have long been of interest to the electric power systems community. An area of particular interest has been \emph{switching problems}; in addition to past work on network reconfiguration, recent work has developed continuous models \cite{agarwal_continuous_switch_2023,babaeinejadsarookolaee_congestion_2023}. This approach, which we refer to as a \emph{continuous admittance} framework, is inclusive of models for transmission line switches and node-breaker substations, and other technologies such as smart wires and unified power flow controllers (UPFC). This framework has seen application for combinatorial switching problems \cite{agarwal_continuous_switch_2023} and current congestion management \cite{babaeinejadsarookolaee_congestion_2023}.
% Both the topology control and the continuous admittance problems in power systems are highly nonlinear and non-convex, often resulting in intractable formulations. 
In the existing literature, different techniques have been investigated to tackle topology control problems, such as convex relaxation techniques (e.g., \cite{low2014convex1, low2014convex2}), linearized approximations (e.g., \cite{ruiz_topology_sensitivities_2007, jabr2019high, bernstein2018load}), machine learning techniques \cite{authier_physics-informed_2024,haider_topology_2025}, and other heuristics \cite{ruiz_topology_heuristics_2011,ruiz_reduced_mip_2012}. These methods are valuable because of their computational efficiency\textemdash in contrast, direct methods for topology control problems typically comprise mixed-integer non-linear programs, which are NP-hard in general and potentially pathological. 

Several simplification techniques are based on distribution factors~\cite{ruiz_topology_sensitivities_2007,ruiz_reduced_mip_2012,guler_generalized_outage_2007,titz2025voltagesensitivedistributionfactorscontingency}, which represent the sensitivity of the power system state to changes in the decision variables. In a typical OPF problem, the power system state refers to the bus voltage phasors, and the decision variables are typically the generators' set-points (or power injections in general). For such applications, the voltage sensitivity coefficients are derived as in~\cite{christakou2013efficient}. In other cases, this theory also has been applied to obtain power losses and current-flow sensitivity factors, which are used to minimize losses and tackle line-current congestion problems~\cite{gupta2020grid, fahmy2021analytical, maharjan2024generalized}. 

Linearization techniques have also been applied to line-switching and topology control problems, such as power transfer distribution factors (PTDFs), line outage distribution factors (LODFs) and outage transfer distribution factors (OTDFs) that are widely used in transmission system applications~\cite{ruiz_reduced_mip_2012,guler_generalized_outage_2007,titz2025voltagesensitivedistributionfactorscontingency}. PTDFs are used for the expressing incremental changes in real power flow that occur on transmission lines due to real power transfers between two regions; these factors are frequently used for pre-contingency analysis by the system operators and can be used to optimize the generator setpoints with the objective to minimize congestion. On the other hand, LODFs are used to determine the impacts of line outages on power flows. Both of these distribution factors often use the DC power flow approximation which is reasonable only for transmission systems.

%Distribution factors are a topic adjacent to our work. 
Recent work in AC line outage distribution factors represents line flow changes with respect to outages while incorporating the full AC power flow equations. The authors of~\cite{yao_novel_ac_factors_2020} developed a method based on the holomorphic embedding load flow (HELF) method. The authors of \cite{guler_generalized_outage_2007} extended the traditional line outage distribution factor computation method to handle an arbitrary number of discrete line outages. 
%%%%%%
The focus of previous work in the literature has been on the impact of changes in power flows due to incremental changes in the power injections. Previous and recent studies have extended these ideas to understanding the sensitivity of power flow solutions to changes in line admittances~\cite{ruiz_topology_sensitivities_2007,ruiz_topology_heuristics_2011,titz2025voltagesensitivedistributionfactorscontingency}, due to their potential for use in various applications in topology control and continuous admittance control. However, the current literature has not yet achieved a formal and complete treatment of this topic, which motivates our own.

Our paper presents a general framework for distribution factors that express the sensitivities of the power system state to changes in the network admittance parameters. These coefficients are computed by implicit differentiation of the power flow equations with respect to the network admittance parameters. The proposed coefficients can be used to produce linearized optimization problems for a number of applications in power transmission and distribution networks alike, including topology control for congestion management and continuous admittance control for voltage regulation.

% \subsubsection{Differentiable optimization}

% \subsubsection{Inverse problems}

% \subsubsection{Topology control}

% \subsubsection{Voltage regulation, congestion management, and other applications}

\textit{Contributions:}
This work presents a generalized technique to analytically differentiate power flow solutions with respect to network admittance parameters. Earlier investigations of this idea, such as the work of~\cite{ruiz_topology_sensitivities_2007,ruiz_topology_heuristics_2011,ruiz_reduced_mip_2012}, show that this technique enables several admittance control applications to be made more efficient. 
%\textcolor{purple}{
Our work contributes an array of new generalizations and extensions of this earlier work, spanning theory, modeling, and applications. In terms of modeling, we consider a broad range of electrical quantities of interest that derive from power flow solutions, generalizing the sensitivity computations to line currents and to arbitrary bus types; in addition, the proposed method is capable of handling shunt admittances, and connections that do not yet exist\textemdash these capabilities are new, to the knowledge of the authors. 
%}

%\textcolor{purple}{
Moreover, these advances open the door to a wide range of applications in improving the scalability of both continuous and discrete topology control and optimization problems. We demonstrate just a few of these applications\textemdash topological impact prediction, voltage control, and congestion management\textemdash in Section~\ref{sec:application}.
%}

% \textcolor{purple}{
Further, we provide additional theoretical contributions that extend the existing literature. Compared to \cite{ruiz_topology_sensitivities_2007,titz2025voltagesensitivedistributionfactorscontingency}, we take the total derivative of the bus injection model of the power flow equations in polar coordinates with respect to network admittance parameters.  Then, by applying the Implicit Function Theorem, 
%} 
we show that the Jacobian of a power flow solution with respect to admittance parameters exists and can be uniquely obtained by solving a linear system of equations, as long as the network is not in a state of static voltage collapse, i.e., the canonical power flow Jacobian is invertible. This mild assumption ensures our method's practical relevance.

% Specifically, we show that the power flow solution-to-admittance Jacobian exists and is unique as long as the network is not in a state of static voltage collapse. This mild assumption ensures the practical relevance of our method. %The method is conceptually similar to recent developments in marginal emissions \cite{valenzuela_dynamic_2024,degleris_emissions_2024} and marginal energy burden factors \cite{west_burden_2024}.

% \textcolor{purple}{
We also present several applications of the proposed method. The method 
provides a continuous approximation of power flow solutions as a function of admittance parameters. This allows for the outcome of the iterative procedure of solving the AC power flow equations with changed admittance to be predicted by taking linearizations about arbitrary admittance parameters of the user's choice. This is particularly useful when performing fast analyses of topology control.
%}

The paper is organized as follows. Section~\ref{sec:nomen_pf_model} reviews the power flow model. Section~\ref{sec:PF_differentiation} presents the main theory for differentiating the power flow solutions with respect to the admittance parameters. Section~\ref{sec:application} demonstrates various applications. Section~\ref{sec:experiments} presents different experiments for these applications. Section~\ref{sec:conclusion} concludes the paper.

% ============================================================
% POWER FLOW MODELS
% ============================================================
\section{Power flow models}
\label{sec:nomen_pf_model}
In this section, we review the bus injection model of the power flow equations. From this, a novel parameterized system of equations is developed that illustrates the conditional linear dependence of bus power injections on the network admittance parameters, which we utilize throughout the work.

\subsection*{Nomenclature}

\textit{Matrix and complex operator notation:}
We denote the sets of real and complex numbers as $\R$ and $\C$, respectively. Given a matrix $\mA$ or a vector $\va$, we denote their transposes as $\mA^\T$  and $\va^\T$, respectively. The $k$-th entry of a vector $\va$ is denoted $(\va)_k$. The operator $\circ$ denotes elementwise multiplication (the Hadamard product). A complex scalar $Z$ with real and imaginary components $X,Y$ is denoted as $Z := X + \uj Y$, where $\uj$ is the imaginary unit, $\uj^2 := -1$. The complex conjugate of scalar $Z=X + \uj Y$ for $X,Y\in\R$ is denoted $\conj{Z} = X - \uj Y$ (analogous definitions apply for complex vectors and matrices). The indicator function is denoted as $\mathds{1}\{ \cdot\}$. For any function $f$ of a complex variable $Z = X + \uj Y$, the Wirtinger derivatives (see \cite{hunger2007introduction}), if they exist, are 
\begin{subequations}
\label{eq:wirtinger-from-components}
\begin{align}
    \frac{\partial f}{\partial Z} &:= \frac{1}{2} \left(\frac{\partial f}{\partial X} - \uj \frac{\partial f}{\partial Y}\right) \\
    \frac{\partial f}{\partial \conj{Z}} &:= \frac{1}{2} \left(\frac{\partial f}{\partial X} + \uj \frac{\partial f}{\partial Y}\right).
\end{align}
\end{subequations}

Wirtinger derivatives will prove useful, as they satisfy the standard rules of differentiation when $Z$ and $\conj{Z}$ are treated as different, independent variables. If the Wirtinger derivatives are known, then the derivatives with respect to the real and imaginary parts can be recovered as
\begin{subequations}
\label{eq:wirtinger-to-components}
\begin{align}
    \frac{\partial f}{\partial X} &= \frac{\partial f}{\partial Z} + \frac{\partial f}{\partial \conj{Z}} \\
    \frac{\partial f}{\partial Y} &= \uj \left(\frac{\partial f}{\partial Z} - \frac{\partial f}{\partial \conj{Z}}\right).
\end{align}
\end{subequations}
Similar relationships hold for complex vector derivatives.

\textit{Network model:} Let~$\setG = (\setN,\setE)$ be an undirected graph corresponding to an electric power system with nodes~$\setN := \{1,\dots,n\}$, and branches~$\setE \subseteq \setN \times \setN$, with~$\setE :\cong \{1,\dots,m\}$. We denote the set of voltage-controlled nodes as~$\setV \subseteq \setN$, and the total number of voltage-controlled nodes as~$g = |\mathcal{V}|$. The set of PQ nodes, where voltage is not controlled, is denoted~$\setP=\setN \setminus \setV$. When ordering the nodes of~$\setP$ and~$\setV$ in ascending order, the~$i$-th node of~$\setV$ is denoted~$\setV(i)$, and the~$i$-th node of~$\setP$ is denoted~$\setP(i)$. If~$k \in \setP$, then~$\setP^{-1}(k)$ denotes the position of $k$ in the list obtained by sorting~$\setP$ in ascending order. We also define~$\setV^{-1}(k)$ in an analogous way.

\textit{Power flow variables and parameters:} Let the state of the network (i.e., the nodal complex voltages in rectangular coordinates) be~$\vx \in \C^n$. In polar coordinates, the voltage magnitudes and phase angles are~$\vv \in \R^n$ and~$\vdelta \in (-\pi,\pi]^n$, respectively; we denote the state as $\vv \angle \vdelta := \vv \circ \exp(\uj \vdelta) \in \C^n$, where $\exp(\cdot)$ is applied elementwise. Thus, $x_i := v_i(\cos(\delta_i) + \uj \sin(\delta_i))$ for all $i \in \setN$. Denote by $\vv_{\setP} \in \R^{(n-g)}$ and $\vv_{\setV} \in \R^g$ the vectors of voltage magnitude at nodes in $\setP$ and $\setV$, respectively. Then the entries of $\vv$ can be written as
\begin{equation}
    (\vv)_k := \begin{cases}
        (\vv_{\setP})_{\setP^{-1}(k)} & k \in \setP \\
        (\vv_{\setV})_{\setV^{-1}(k)} & k \in \setV.
    \end{cases}
\end{equation}
As the voltages in $\vv_{\setV}$ are specified, we only need to solve the power flow equations for $\vv_{\setP}$ and $\vdelta$. It is sometimes convenient to write $\vv$ as a function of $\vv_{\setP}$, that is $\vv = \vv(\vv_{\setP})$, as a way to emphasize that only the entries of $\vv$ associated with nodes in $\setP$ are considered variables in the power flow equations. 

Branch current flows in rectangular coordinates are denoted as $\vu \in \C^m$. In polar coordinates, let~$\vell \in \R^m$ and~$\vphi \in (-\pi,\pi]^m$ denote the branch current magnitudes and phase angles, respectively. Thus,~$\vu := \vell \circ \exp(\uj \vphi)$ and~$u_{ij} := \ell_{ij} (\cos(\phi_{ij}) + \uj \sin(\phi_{ij}))$ for all~$(i,j) \in \setE$.

% SADR: I think this is an ideal place to add a line diagram.
\begin{figure}[h]
    \begin{center}        
        \begin{tikzpicture}[thick, american voltages]

    %    \draw[help lines];
        \node[coordinate] (I1) at (0, 0){}; 
        \node[coordinate] (I2) at (0, -2){};
    
        % SADR: sending end (node i)
        \draw (I1) to [open, v=$$,] (I2);
        % SADR: actual label
        \node at (0.5, -1) {\footnotesize${x_{i}=v_i \exp (\uj\delta_i)}$};
        
        % SADR: input current to the line
        \draw (I1) -- ++(2,0) node[coordinate, name=Vi]{};
        \draw (Vi) to [generic={\footnotesize $y_{i}$}] ++(0, -2);
        \draw (Vi) to [generic, l^={\footnotesize$y_{ij}$}, i={$$}] ++(3, 0) node[coordinate, name=Vj]{};
        % SADR: actual label.
        \node at (5.15, 0.30) {\footnotesize $u_{ij} = \ell_{ij} \cos \phi_{ij}$};
        \draw (Vj) to [generic={\footnotesize $y_{j}$}] ++(0, -2);
    
        \draw (Vj) to [short] ++(2, 0) node[coordinate, name=O1]{};
        \draw (O1) to [open, v=$$] ++(0, -2) node[coordinate, name=O2]{};
        \node at (6.85, -1) {\footnotesize $x_j = v_j \exp (\uj\delta_j)$};
        \draw (I2) to [short] (O2);
    
\end{tikzpicture}
    \end{center}
    \caption{Illustration of the $\pi$-line model and the line parametrization used in our work.}
\end{figure}
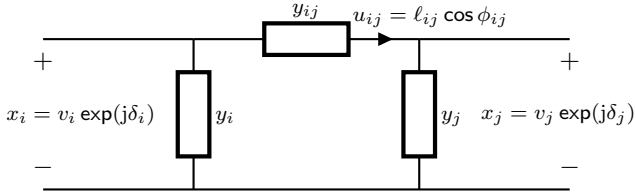

We denote the nodal admittance matrix as $\mY := \mG + \uj \mB \in \C^{n \times n}$. Let $\vs := \vp + \uj \vq \in \C^n$ denote the net complex power injections at each bus. The power flow equations are then
\begin{equation}
    \vs = \diag(\vx) \conj{\mY \vx},
\end{equation}
where $\conj{(\cdot)}$ denotes the complex conjugate and $\diag(\cdot)$ denotes a diagonal matrix with the argument along the diagonal.

\subsection{Admittance-parameterized power flow equations}

Given a power system model $\setG = \left(\setN, \setE \right)$, consider the corresponding complete graph $K_{n} = \left(\setN, \setE_{K} \right)$ where $\setE_{K}$ has all possible connections between the nodes in $\setN$, i.e., $\binom{n}{2} = n(n-1)/2$ edges. Let $\setS \subseteq \setE_{K}$ be the set of all \emph{controllable lines}. We understand the set $\setS$ in a very general sense, where \emph{controllable includes switchable}. Specifically, $\setS$ contains all connections between nodes $i,j \in \setN$ where
\begin{enumerate}
    \item a line admittance can be modified, e.g., lines with ``smart wires" \cite{kreikebaum_smart_wires_2010}, FACTS compensation, or switches,
    \item a line could potentially be created, but does not yet exist, i.e., $\left(i, j \right) \in \setE_{K}$ but $\left(i, j \right) \notin \setE$,
    \item a line that could be switched open or closed.
\end{enumerate}
Therefore, in our work, we will consider the admittances of each line $(i,j) \in \setE_{K}$ to be parameterized functions $y_{ij} : \C \to \C$ of the form 
\begin{equation}
\label{eq:paramterization-of-admittance}
    y_{ij}(\gamma_{ij}) := \begin{cases}
        \gamma_{ij}y_{ij}^{\bullet} & (i,j) \in \setS\\
        y_{ij}^\bullet & (i,j) \in \setE_{K} \setminus \setS\\
        0 & (i,j) \notin \setE
    \end{cases} ,
\end{equation}
where $y_{ij}^\bullet \in \C$ is a chosen nominal admittance value. 

Without loss of generality, we will consider the case where~$\abs{\gamma_{ij}} \leq 1$; in principle, the maximum values of~$\abs{\gamma_{ij}}$ could be larger, such as in the case of controllable smart wires. In the special case where~$\Im(\gamma_{ij}) = 0$ and~$\Re(\gamma_{ij}) \in [0,1]$,  the parameterization~\eqref{eq:paramterization-of-admittance} describes a switching process, where~$\gamma_{ij} = 0 + \uj 0$ if the line is switched open. 

With the above formulation, we can now represent a \emph{controllable admittance matrix} in a vectorized form $\vw : \C^{m+n} \to \C^{n(n+1)/2}$ whose entries correspond to the admittance parameters, including self-admittance. Suppressing the $\vgamma$ argument for convenience, we construct the vector $\vw$ as
    \begin{equation}
    \label{eq:def:all-possible-weights}
     \vw := \begin{bmatrix}
        y_{ij \: : \: (i,j) \in \setE_{K}} & y_{i \: : \: i \in \setN}
    \end{bmatrix}^\T.
\end{equation}
Critically, the vector $\vw$ in \eqref{eq:def:all-possible-weights} has been constructed independently of $\setS$\textemdash it is the vector of \textit{all possible network parameters}, regardless of the network topology. %Note that the dimension of $\vw$ is identified with the power set of $\setN$ because this is the most general case. 
It is straightforward for the domain of $\vw$ to be constrained to an appropriate lower-dimensional subspace corresponding to a subset of the ${n \choose 2}$ possible lines under practical reconfiguration and parameter constraints. Using the most general formulation, we obtain the following result.
\begin{lemma}
\label{lemma:special-pf-eqns}
    % There exists a function $\vs : \C^n \times \C^{n + m} \to \C^n$ such that the power flow equations can be represented as a parameterized function of the admittance parameters,
    Let~$\vx \in \C^n$ be a nodal voltage state inducing power injections~$\vs \in \C^n$. There exists a matrix-valued function~$\mF: \C^n \to \C^{n \times n(n-1)/2}$ such that the complex conjugate of the power injections can be written as a parameterized function of all possible admittance weights~$\vw \in \C^{{n \choose 2}}$, namely,
    \begin{equation}
    \label{eq:topology-paramterized-pf-eqns}
        \conj{\vs}(\vx;\vw) = \mF(\vx) \vw,
    \end{equation}
    where $\mF: \C^n \to \C^{n \times n(n-1)/2}$ is defined explicitly in \eqref{eq:special-pf-eq-representation}, and $\vw$ is as in \eqref{eq:def:all-possible-weights}.
\end{lemma}
\begin{proof}
First we prove the claim for the case where the system has no phase-shifting transformers. 

Note that for any number of buses $n$, there exists a mapping~$\loA_{(\cdot)} : n \mapsto \{-1,0,1\}^{n(n-1)/2 \times n}$ to the unique incidence matrix of a complete undirected graph with a self-edge at every node. Denote the $i$-th standard basis vector in~$\R^d$ as~$\ve^{d}_i \in \R^d$; we then have that~$\loA$ takes the form
\begin{equation}
    \label{eq:def:agnostic-incidence-matrix}
    \begin{split}
    \loA_n &= 
    % \begin{bmatrix}
    % \vone_{n-1} &  -\mId_{n-1} & \cdot & \cdot & \cdot \\
    % \vzero_{n-2} & \vone_{n-2} & -\mId_{n-2} & \cdot & \cdot \\
    % \vdots & \ddots & \ddots & \ddots & \cdot \\
    % 0 & \dots & 0 & 1 & -1\\
    % % \vzero_{n-3} & \vzero_{n-3} & \vone_{n-3} &  -\mId_{n-3} &  \cdot\\
    % % \vdots & \cdot & \ddots & \ddots & \ddots \\
    % % 0 &  \dots & \dots & 0 & 1 & -1\\
    % \hline
    %  \cdot & \cdot & \mId_{n} & \cdot & \cdot \\
    % \end{bmatrix},
    % = 
    \begin{bmatrix}
    \vone_{n-1} & -\ve^{n-1}_{1} & -\ve^{n-1}_{2}  & \dots  & -\ve^{n-1}_{n-1}\\
    \vzero_{n-2} & \vone_{n-2} & -\ve^{n-2}_{1} &\ddots & -\ve^{n-2}_{n-2}\\
    \vdots & \ddots & \ddots & \ddots & \vdots \\
    % 0 & \dots & \vzero_{2} & \vone_2 & -\ve_{2,2}\\
    0 & \dots & 0 & 1 & -1 \\
    \hline
     \cdot & \cdot & \mId_{n} & \cdot & \cdot \\
    \end{bmatrix},
    %\\
    % &= \begin{bmatrix}
    %     \mA(n)\\
    %     \mId_n
    % \end{bmatrix}
    \end{split}
\end{equation}
where $\vone_d$ and $\vzero_d$ denote $d$-dimensional vectors of all ones and zeros, respectively, and $\mId_d$ denotes an identity matrix of dimension $d \times d$. The $\mId_n$ block concatenated at the bottom of \eqref{eq:def:agnostic-incidence-matrix} is due to the presence of self edges at every node. We remark that for \textit{any possible topology weights} $\vw$, we have
\begin{subequations}
\label{eq:special-pf-eq-representation}
    \begin{align}
        \conj{\vs} = \diag(\conj{\vx}) Y\vx &= \diag(\conj{\vx}) \loA_n^\T \diag(\vw) \loA_n \vx\\
        &= \underbrace{\diag(\conj{\vx}) \loA_n^\T \diag\left(\loA_n \vx\right)}_{:= \mF(\vx)} \vw\\
        &:= \mF(\vx) \vw,
    \end{align}
\end{subequations}
which yields the desired representation of the power flow equations. 
% where $F: \C^n \to \C^{n \times n+m}$ with $F(x) := \diag(x)^* A^\T \diag(Ax)$. Note that simply taking the complex conjugate of the above recovers the original measurements.

The case where the system has phase-shifting transformers is proven in the same way, with the only difference that now $\loA_n$ has the form
\begin{equation}
    \loA_n = \begin{bmatrix}
        \mA(n)\\
        \mId
    \end{bmatrix}.
\end{equation}
where $\mA(n)$ is the \textit{generalized incidence matrix} (see \cite{turizo2023}) of a complete $n$-node graph without self edges, and with phase-shifts chosen to match those of the system transformers.
\end{proof}

    Considering the space of all possible admittance parameters\textemdash as in Lemma \ref{lemma:special-pf-eqns}\textemdash may not always be necessary, except in the most unrestricted planning contexts. In the context of operational reconfiguration, the controllable topological components are typically a small subset of all components in the network. 
It is straightforward to generalize the result of Lemma~\ref{lemma:special-pf-eqns} to handle this setting. Suppose that a subset~$\setS \subseteq \setE$ of possible connections can exist, collecting a total of~$\abs{\setS}$ lines. Introduce a line selection matrix~$\mE \in \c{0,1}^{\abs{\setS} \times {n \choose 2}} $ where~$E_{i,e} = 1$ if~$e = (i,j) \in \setS$, $j>i$. The selection matrix~$\mE$ encodes the edges of the complete graph over which we have control. Reloading the notation, we can redefine the incidence matrix operator $\loA_{\p{\cdot}} : n \mapsto \c{-1,0,1}^{\p{\abs{\setS}+n} \times n}$ as 
\[
\loA_n := \begin{bmatrix}
    \mE\mA(n) \\
    \mId
\end{bmatrix},
\]
and proceed with a computation analogous to \eqref{eq:special-pf-eq-representation}.

% ============================================================
% DIFFERENTIATION OF THE POWER FLOW SOLUTIONS WITH RESPECT TO ADMITTANCES
% ============================================================
\section{Differentiating power flow solutions with respect to admittance parameters}
\label{sec:PF_differentiation}
This section presents the proposed method to derive the sensitivity coefficients of nodal voltage and current phasors with respect to the admittance parameters. We show how they can be derived by implicit differentiation. The derivations for the voltages and the currents are presented as follows.
%%%%
\subsection{Nodal voltages}
\label{sec:diff-v-wrt-w}
The primary assumption upon which the analytical results are based is stated below.
\begin{assumption}
    \label{assum:no-voltage-collapse}
    Let $(\vv^\bullet_\setP,\vdelta^\bullet) \in \R^{(n-g)}_+ \times (-\pi,\pi]^n$ and $\vx^\bullet \in \C^n$ be a solution to the power flow equations in polar and rectangular coordinates, respectively, where  $\vx^\bullet := \vv(\vv^\bullet_\setP) \circ \exp(\uj \vdelta^\bullet)$. The corresponding power flow Jacobian is assumed to be non-singular at such an operating point.
\end{assumption}
Importantly, Assumption~\ref{assum:no-voltage-collapse} implies that network is not in a state of voltage collapse. This is not a restrictive assumption, as power systems usually operate far from this condition; however, such states can occur, see~\cite{cutsem_reactive_power_voltage_collapse_1991,hiskens_exploring_2001,simpson-porco_voltage_2016} for further exposition on this topic.

\begin{theorem}
\label{thm:dvdw}
    Consider a power system with admittance parameters~$\vw^\bullet := \vg^\bullet + \uj \vb^\bullet$ and a corresponding power flow solution~$(\vv^\bullet_\setP,\vdelta^\bullet)$. If Assumption~\ref{assum:no-voltage-collapse} holds, the Jacobian of the power flow solution with respect to the network admittance parameters~$\vw := \vg + \uj \vb$ is well-defined in a neighborhood of~$(\vg^\bullet, \vb^\bullet)$. %unique voltage-topology sensitivities $\frac{\partial v}{\partial b}$ and $\frac{\partial v}{\partial g}$.
\end{theorem} 
\begin{proof}
Denote the real and imaginary components of the vectorized admittance weights as $ \vg := \Re\{\vw\}$ and $ \vb := \Im\{ \vw\}$. Then, by applying Lemma \ref{lemma:special-pf-eqns}, we can construct the topology-parameterized power flow equations \eqref{eq:topology-paramterized-pf-eqns} in terms of polar coordinates for the state as the function 
\begin{equation}
    \conj{\vs} := \mF(\vv_\setP,\vdelta)\vw := \mF(\vv(\vv_\setP) \circ \exp(\uj \vdelta)) (\vg + \uj \vb).
\end{equation}
There is a mapping $\vkappa : \R^n \times \R^n \times \R^{n(n-1)/2} \times \R^{n(n-1)/2} \to \R^{(2n-g)}$ parameterized by the network topology parameters $\vg,\vb$ with entries of the form
\begin{align*}
    &(\vkappa(\vv_\setP,\vdelta;\vg,\vb))_k = \nonumber \\ &\hspace*{5em}\begin{cases}
    \left(\Re\{\mF(\vv_\setP,\vdelta)(\vg + \uj \vb)\} - \vp^\bullet\right)_k & k \leq n \\
    \left(\Im\{\mF(\vv_\setP,\vdelta)(\vg + \uj \vb)\} + \vq^\bullet\right)_{\setP(k-n)} & k > n \label{eq:fixed-point} \\
    \end{cases},
\end{align*} 
such that~$\vkappa(\vv^\bullet_\setP,\vdelta^\bullet;\vg^\bullet,\vb^\bullet) = \vzero_{2n-g}$. Let~$\partial \vkappa$ denote the derivative of~$\vkappa$ with respect to the power system state at~$(\vv^\bullet_\setP,\vdelta^\bullet;\vg^\bullet,\vb^\bullet)$, then
\begin{equation}
    \partial \vkappa = \begin{bmatrix}
    \frac{\partial \vkappa}{\partial \vdelta} &  \frac{\partial \vkappa}{\partial \vv_\setP} 
\end{bmatrix} = \mJ(\vv^\bullet_\setP,\vdelta^\bullet),
\end{equation}
where $\mJ(\vv_\setP,\vdelta)$ is the standard power flow Jacobian matrix in polar coordinates. If Assumption \ref{assum:no-voltage-collapse} holds, then $\mJ(\vv^\bullet_\setP,\vdelta^\bullet)$ is non-singular, so $\partial \vkappa$ is non-singular as well. By the Implicit Function Theorem, there exists unique, continuously differentiable functions $\vv_\setP(\vg,\vb)$ and $\vdelta(\vg,\vb)$ defined in a neighborhood of $(\vg^\bullet,\vb^\bullet)$ such that $\vv_\setP(\vg^\bullet,\vb^\bullet) = \vv_\setP^\bullet$, $\vdelta(\vg^\bullet,\vb^\bullet) = \vdelta^\bullet$, and $\vkappa(\vv_\setP(\vg,\vb),\vdelta(\vg,\vb);\vg,\vb) = \vzero_{2n-g}$ for all $(\vg, \vb)$ in the neighborhood. Moreover, the derivatives of these functions satisfy
\begin{subequations}
\begin{align}
     \begin{bmatrix}
        \frac{\partial \vdelta}{\partial \vg} & \frac{\partial \vdelta}{\partial \vb}\\
        \frac{\partial \vv_\setP}{\partial \vg} & \frac{\partial \vv_\setP}{\partial \vb}
    \end{bmatrix} &= -(\partial \vkappa)^{-1}\begin{bmatrix}
        \frac{\partial \vkappa}{\partial \vg} & \frac{\partial \vkappa}{\partial \vb}
    \end{bmatrix} \\
    % \begin{bmatrix}
    %     \frac{\partial \vdelta}{\partial \vg} & \frac{\partial \vdelta}{\partial \vb}\\
    %     \frac{\partial \vv_\setP}{\partial \vg} & \frac{\partial \vv_\setP}{\partial \vb}
    % \end{bmatrix} 
    &= -\mJ(\vv^\bullet_\setP,\vdelta^\bullet)^{-1}\begin{bmatrix}
        \frac{\partial \vkappa}{\partial \vg} & \frac{\partial \vkappa}{\partial \vb}
    \end{bmatrix}, \label{eq:implicit-function-theorem-sol}
\end{align}
\end{subequations}
so by construction the claim holds.
\end{proof}
In the expression \eqref{eq:implicit-function-theorem-sol}, the Jacobian matrices taken with respect to the admittance parameters can be analytically computed from \eqref{eq:special-pf-eq-representation}; see Appendix \ref{apdx:diff-power-to-params}. Furthermore, from \eqref{eq:wirtinger-from-components} we can compute the Wirtinger derivatives of the power flow solution with respect to the network parameters $\vw$ as
\begin{subequations}
\begin{align}
    \frac{\partial \vdelta}{\partial \vw} & = \frac{1}{2} \left(\frac{\partial \vdelta}{\partial \vg} - \uj \frac{\partial \vdelta}{\partial \vb}\right), \\
    \frac{\partial \vdelta}{\partial \conj{\vw}} & = \frac{1}{2} \left(\frac{\partial \vdelta}{\partial \vg} + \uj \frac{\partial \vdelta}{\partial \vb}\right), \\
    \frac{\partial \vv_\setP}{\partial \vw} & = \frac{1}{2} \left(\frac{\partial \vv_\setP}{\partial \vg} - \uj \frac{\partial \vv_\setP}{\partial \vb}\right), \\
    \frac{\partial \vv_\setP}{\partial \conj{\vw}} & = \frac{1}{2} \left(\frac{\partial \vv_\setP}{\partial \vg} + \uj \frac{\partial \vv_\setP}{\partial \vb}\right).
\end{align}
\end{subequations}

Theorem \ref{thm:dvdw} shows that the expressions \eqref{eq:p-to-w-derivatives} and \eqref{eq:q-to-w-derivatives}, coupled with the power flow Jacobian $\mJ(\cdot,\cdot)$, provide all information needed to recover the voltage-admittance sensitivities by solving the system \eqref{eq:implicit-function-theorem-sol} for a given operating point.

\subsection{Line current flows}
\label{sec:diff-curr-wrt-w}

Now, we establish how to differentiate the line currents with respect to the admittance parameters. 

\subsubsection{Shuntless networks}
When there are no shunts in the network, by Ohm's law, the line current phasors are 
\begin{equation}
    u_{ij} = (g_{ij} + \uj b_{ij})(x_i - x_j), \quad \forall \;(i,j) \in \setE.
\end{equation}
Now, define the voltage phasor difference~$\Delta_{ij} := x_i - x_j \in \C^n$ across a line~$(i,j) \in \setE$. Recall that for any continuously differentiable function~$\vf : \R^n \to \C^m$, we have that~$\frac{\partial}{\partial \vx}\vf(\vx)  = \frac{\partial}{\partial \vx} \Re\{\vf(\vx)\} + \uj \frac{\partial}{\partial \vx} \Im\{\vf(\vx)\}$. Consequently, by the above identities and the chain and product rules, the derivatives of the line currents with respect to the network conductance parameters are
\begin{equation}
\label{eq:dldg}
    \frac{\partial u_{ij}}{\partial g_{kl}} = \begin{cases}
        \frac{\partial \Delta_{ij}}{\partial g_{ij}}g_{ij} + \Delta_{ij} + \uj\, \frac{\partial \Delta_{ij}}{\partial g_{ij}} b_{ij} & (i,j)=(k,l)\\
        \frac{\partial \Delta_{ij}}{\partial g_{kl}}g_{ij} + \uj\, \frac{\partial \Delta_{ij}}{\partial g_{kl}} b_{ij}  & \ow
    \end{cases}
\end{equation}
and similarly, the derivatives with respect to the susceptance parameters are 
\begin{equation}
\label{eq:dldb}
    \frac{\partial u_{ij}}{\partial b_{kl}} = \begin{cases}
        \frac{\partial \Delta_{ij}}{\partial b_{ij}}g_{ij}  + \uj\left( \frac{\partial \Delta_{ij}}{\partial b_{ij}} b_{ij} + \Delta_{ij} \right) & (i,j)=(k,l)\\
        \frac{\partial \Delta_{ij}}{\partial b_{kl}}g_{ij} + \uj \frac{\partial \Delta_{ij}}{\partial b_{kl}} b_{ij} & \ow
    \end{cases}
\end{equation}
for all~$(i,j)\in \setE_{K}, i \neq j$ and~$(k,l) \in \setE_{K}, k \neq l$. Notice that~$u_{ij} = -u_{ji}$, so we only need to compute sensitivities for currents in one direction.

\subsubsection{Generalization to networks with shunts}
In the case where there are shunts in the network, let~$\vu^{\sf br} \in \C^{n(n-1)/2}$ and~$\vu^{\sf sh} \in \C^n$ denote the branch and shunt current phasors in rectangular coordinates, respectively. Let~$\vu \in \C^{n(n-1)/2}$ be the concatenation of these two vectors. The branch and shunt currents can be computed from \eqref{eq:special-pf-eq-representation} as
\begin{subequations}
 \label{eq:thm-line-curr-with-shunts}
\begin{align}
    \vu :=\begin{bmatrix}
        \vu^{\sf br}\\
        \vu^{\sf sh}
    \end{bmatrix} &= \diag(\vw) \loA_n \vx\\
    &= \diag(\vw)\begin{bmatrix}
        \mA(n)\\
        \mId
    \end{bmatrix} \vx \\
    &= \diag\left(\begin{bmatrix}
        \mA(n)\\
        \mId
    \end{bmatrix} \vx\right) \vw,
\end{align}
\end{subequations}
where $\mA(n)$ is the generalized incidence matrix of a complete graph with $n$ nodes and phase shifts matching that of the system. Thus, differentiating, we obtain the block matrix
\begin{equation}
    \frac{\partial \vu}{\partial \vw} = \begin{bmatrix}
        \dfrac{\partial \vu^{\sf br}}{\partial \vw} &
        \dfrac{\partial \vu^{\sf sh}}{\partial \vw}
    \end{bmatrix}^{\top},
\end{equation}
where the blocks 
\begin{subequations}
\label{eq:current_sensitivities}
\begin{align}
\frac{\partial \vu^{\sf br}}{\partial \vw} &= \diag\left(\mA(n) \vx\right) + \diag(\vw^{\sf br})\mA(n) \frac{\partial \vx}{\partial \vw}, \\
\frac{\partial \vu^{\sf sh}}{\partial \vw} &= \diag\left(\vx\right) + \diag(\vw^{\sf sh}) \frac{\partial \vx}{\partial \vw},
\end{align}
\end{subequations}
where $\vw^{\sf br}$ is the vector formed by the first $n(n-1)/2$ entries of $\vw$ (the branch weights) and $\vw^{\sf sh}$ is the vector formed by the last $n$ entries of $\vw$ (the shunt weights). Similarly, differentiating with respect to the complex conjugate of the admittance vector, we arrive at
\begin{equation}
    \frac{\partial \vu}{\partial \conj{\vw}} = \begin{bmatrix}
        \dfrac{\partial \vu^{\sf br}}{\partial \conj{\vw}} &
        \dfrac{\partial \vu^{\sf sh}}{\partial \conj{\vw}}
    \end{bmatrix}^{\top},
\end{equation}
where the blocks are computed as
\begin{subequations}
\begin{align}
\frac{\partial \vu^{\sf br}}{\partial \conj{\vw}} &= \diag(\vw^{\sf br})\mA(n) \frac{\partial \vx}{\partial \conj{\vw}}, \\
\frac{\partial \vu^{\sf sh}}{\partial \conj{\vw}} &= \diag(\vw^{\sf sh}) \frac{\partial \vx}{\partial \conj{\vw}}.
\end{align}
\end{subequations}

Lastly, we derive the sensitivities of the current flow magnitudes~$\vell = |\vu|$. To this end, observe that
\begin{equation}
    \vell = \left(\diag(\vu) \conj{\vu}\right)^{1/2},
\end{equation}
so the sensitivities are computed as
\begin{subequations}
\label{eq:current_magnitude_sens}
\small
\begin{align}
\frac{\partial \vell}{\partial \vw} &= \frac{1}{2} \diag(\vell)^{-1} \left(\diag(\conj{\vu})\frac{\partial \vu}{\partial \vw} + \diag(\vu) \conj{\left(\frac{\partial \vu}{\partial \conj{\vw}}\right)}\right), \\
\frac{\partial \vell}{\partial \conj{\vw}} &= \frac{1}{2} \diag(\vell)^{-1} \left(\diag(\conj{\vu})\frac{\partial \vu}{\partial \conj{\vw}} + \diag(\vu) \conj{\left(\frac{\partial \vu}{\partial \vw}\right)}\right).
\end{align}
\end{subequations}

\subsection{Line power flows}
\label{sec:diff-powerflow-wrt-w}
In Sections~\ref{sec:diff_powerflow_wrt_w_polar} and~\ref{sec:diff_powerflow_wrt_w_rectangular}, the derivatives of the line power flows with respect to arbitrary controllable admittances are computed for lines \textit{without} shunts, both polar and rectangular form, respectively. Subsequently, we will generalize the result to handle networks with shunts. Then, in Section \ref{sec:diff_powerflow_wrt_w_shunts}, the results are generalized to lines \textit{with} shunt admittances.

\subsubsection{Polar form}
\label{sec:diff_powerflow_wrt_w_polar}
Let $w_{ij} := \abs{w_{ij}} \exp \left(\uj \phi_{ij}\right)$ denote the polar form of the admittance weight for an arbitrary line $(i,j) \in \setE_{K}$, where $\abs{w_{ij}} = \sqrt{g_{ij}^2 + b_{ij}^2}$ and $\phi_{ij} = \operatorname{atan2}\p{b_{ij}, g_{ij}}$. The active and reactive power flows across the line are
\begin{subequations}
    \label{eq:complex_power_flow}
    \begin{align}
        p_{ij} &:= v_i \abs{w_{ij}} \p{ v_i \cos\p{\phi_{ij}} - v_j \cos\p{\delta_{ij} - \phi_{ij} } }, \\
        q_{ij} &:= - v_i \abs{w_{ij}} \p{ v_i \sin\p{ \phi_{ij}} + v_j \sin \p{\delta_{ij} - \phi_{ij} } }.
    \end{align}
\end{subequations}

\subsubsection{Rectangular form}
\label{sec:diff_powerflow_wrt_w_rectangular}
The rectangular form follows immediately from the previous line current flow sensitivity derivation. Recall that grid physics requires that the line current flows are given in rectangular coordinates as
\[
s_{ij} := x_i \conj{u_{ij}} \quad \forall (i,j) \in \setE_{K}, \, i \neq j.
\]

Hence, the sensitivities of the complex flow $s_{ij} \in \C$ through line $(i,j)$ with respect to the conductance and susceptance weights of an arbitrary line $(k,l) \in \setE_{k}, \, k \neq l$ are
\begin{subequations}
\begin{align}
    \frac{\partial s_{ij}}{\partial g_{kl}} &= \frac{\partial x_i}{\partial g_{kl}} \conj{u_{ij}} + x_i \conj{\left(\frac{\partial u_{ij}}{\partial g_{kl}}\right)}, \\
    \frac{\partial s_{ij}}{\partial b_{kl}} &= \frac{\partial x_i}{\partial b_{kl}} \conj{u_{ij}} + x_i \conj{\left(\frac{\partial u_{ij}}{\partial b_{kl}}\right)},
\end{align}
\end{subequations}
respectively. Critically, notice that $s_{ij}$ may not necessarily be equal to $s_{ji}$, so we need to compute sensitivities for power flows \textit{in both directions}.

\subsubsection{Generalization to networks with shunts }
\label{sec:diff_powerflow_wrt_w_shunts}
Next we consider the case where there are shunts in the network. Let~$\vs^{\sf brd} \in \C^{n(n-1)/2}$ be the branch complex power flows, in the direction implied by~$\mA(n)$. We also define~$\vs^{\sf sh} \in \C^n$ denote the shunt complex power flows. Lastly, we let~$\vs^{\sf brr} \in \C^{n(n-1)/2}$ be the branch complex power flows in the reverse direction. Let~$\vs^{\sf fl} \in \C^{n^2}$ be the concatenation of these three vectors, then the branch and shunt currents can be calculated from~\eqref{eq:special-pf-eq-representation} as
\begin{equation}
\label{eq:line-flows-with-shunts}
    \vs^{\sf fl} :=\begin{bmatrix}
        \vs^{\sf brd}\\
        \vs^{\sf sh}\\
        \vs^{\sf brr}
    \end{bmatrix} = \diag(\loC_n \vx) \conj{\begin{bmatrix}
        \vu\\
        -\vu^{\sf br}
    \end{bmatrix}}.
\end{equation}
The matrix $\loC_n \in \R^{n^2 \times n}$ is defined as
\begin{equation}
    \label{eq:def:c-matrix}
    \loC_n := 
    \begin{bmatrix}
    \vone_{n-1} & \vzero_{n-1} & \vzero_{n-1}  & \dots  & \vzero_{n-1}\\
    \vzero_{n-2} & \vone_{n-2} & \vzero_{n-1} &\ddots & \vzero_{n-1}\\
    \vdots & \ddots & \ddots & \ddots & \vdots \\
    0 & \dots & 0 & 1 & 0 \\
    \hline
    \cdot & \cdot & \mId_{n} & \cdot & \cdot \\
    \hline
    \vzero_{n-1} & \ve^{n-1}_{1} & \ve^{n-1}_{2}  & \dots  & \ve^{n-1}_{n-1}\\
    \vzero_{n-2} & \vzero_{n-2} & \ve^{n-2}_{1} &\ddots & \ve^{n-2}_{n-2}\\
    \vdots & \ddots & \ddots & \ddots & \vdots \\
    0 & \dots & 0 & 0 & 1
    \end{bmatrix},
\end{equation}
where $\vone_d,\vzero_d$ denote $d$-dimensional vectors of all ones and zeros, respectively, and $\ve^{d}_i \in \R^d$ denotes the $i$-th standard basis vector of $\R^d$. Differentiating these expressions yields
\begin{subequations}
\label{eq:line_flow_sensitivities}
\small
\begin{align}
\frac{\partial \vs^{\sf fl}}{\partial \vw} &= \diag\left(\begin{bmatrix}
        \conj{\vu}\\
        -\conj{\vu}^{\sf br}
    \end{bmatrix}\right) \loC_n \frac{\partial \vx}{\partial \vw} + \diag(\loC_n \vx) \conj{\begin{bmatrix}
        \frac{\partial \vu}{\partial \conj{\vw}}\\
        -\frac{\partial \vu^{\sf br}}{\partial \conj{\vw}}
    \end{bmatrix}}, \\
\frac{\partial \vs^{\sf fl}}{\partial \conj{\vw}} &= \diag\left(\begin{bmatrix}
        \conj{\vu}\\
        -\conj{\vu}^{\sf br}
    \end{bmatrix}\right) \loC_n \frac{\partial \vx}{\partial \conj{\vw}} + \diag(\loC_n \vx) \conj{\begin{bmatrix}
        \frac{\partial \vu}{\partial \vw}\\
        -\frac{\partial \vu^{\sf br}}{\partial \vw}
    \end{bmatrix}}.
\end{align}
\end{subequations}

% ============================================================
% APPLICATIONS
% ============================================================
\section{Application of Admittance Sensitivities}
\label{sec:application}
In this section, we illustrate two distinct applications of the proposed power flow linearization technique. The first is on estimating the power flow solutions with change in topology without the need of solving the power flow equations directly. The second consists of controlling the admittance parameters continuously such as line impedances for voltage regulation.
\subsection*{Predicting power flow solution with a change in topology}
\label{sec:application-linearization}
The voltage-admittance sensitivity matrices allow us to linearize a power flow solution around nominal admittance parameters $\vg^\bullet$ and $\vb^\bullet$. For example, we can form functions that approximate the voltage magnitude component of a power flow solution $\vvhat$ and the corresponding branch current magnitudes $\vellhat$ as a function of the topology parameters, where
\begin{equation}
\label{eq:sensitivity-definition}
         \begin{bmatrix}
             \vvhat(\vg,\vb)\\
             \vellhat(\vg,\vb)
         \end{bmatrix} = \begin{bmatrix}
             \vv^\bullet_\setP \\
             \vell^\bullet
         \end{bmatrix}+ 
    %      \begin{bmatrix}
    %     \frac{\partial \vv}{\partial \vg} & \frac{\partial \vv}{\partial \vb}\\
    %     \frac{\partial \vell}{\partial \vg} & \frac{\partial \vell}{\partial \vb}
    % \end{bmatrix} 
     \begin{bmatrix}
             \mK^v_g & \mK^v_b\\
             \mK^\ell_g & \mK^\ell_b\\
        \end{bmatrix}
    \begin{bmatrix}
        \vg - \vg^\bullet\\
        \vb - \vb^\bullet
    \end{bmatrix}.
\end{equation}
where~$\mK^v_g = \partial \vv_\setP / \partial \vg$,~$\mK^v_b = \partial \vv_\setP / \partial \vb$,~$\mK^\ell_g=\partial \ell / \partial \vg$, and~$\mK^\ell_b = \partial \ell / \partial \vb$. For the line flows, we may consider only the flows of lines in $\setE$ instead of $\setE_{K}$. These coefficients are computed using \eqref{eq:implicit-function-theorem-sol} and \eqref{eq:current_sensitivities} for a pre-defined operating point around which power flow equations are linearized. Later in the numerical experiments in Section~\ref{sec:experiments}, we compare the estimation performance with different operating points. 

\subsection*{Admittance control}
\label{sec:application-continuous-admittance} 
Continuous admittance control has been proposed in the existing literature such as \emph{smart wires}, which allow for induced changes in line admittances~\cite{jinsiwale_decentralized_2020} and shunt reactance control~\cite{peelo1996new, xiao2003available, sahraei2015day}. Below, we show a formulation using the proposed power flow linearization for carrying out admittance control for voltage regulation. 
\subsubsection*{Voltage regulation with variable admittance}
\label{sec:application-vreg}
A simple application of the proposed framework is in allowing admittance parameters to directly become a decision variable in linearized voltage regulation problems. Traditionally, this is solved via heuristic approaches; in FACTS problems, the controllable admittance is modeled as the reactive power injection at a PQ bus, i.e., $Q_i = v_i^2Y_{ij}^{\sf sh}$.

We demonstrate the practical application of these coefficients via a simple voltage control problem. In particular, applying the linearization approach from Section \ref{sec:application-linearization}, we can construct a linear approximation of the voltage magnitudes and currents around both active and reactive power injections via the matrix
\begin{equation}
   \mK := \begin{bmatrix}
            \mK^v_p & \mK^v_q & \mK^v_g & \mK^v_b\\
            \mK^\ell_p & \mK^\ell_q & \mK^\ell_g & \mK^\ell_b\\
        \end{bmatrix}.
\end{equation}
Define the overall state of the network as 
\begin{equation}
    \vz := \begin{bmatrix}
        \vp^\T & \vq^\T_\setP & \vg^\T & \vb^\T
    \end{bmatrix}^\T,
\end{equation}
and let $\vz^\bullet$ be the nominal values of these variables. Define their lower and upper bounds as
\begin{equation}
    \vz^{\sf min} := \begin{bmatrix}
            \vp^{\sf min}\\
            \vq^{\sf min}_\setP\\
            \vg^{\sf min}\\
            \vb^{\sf min}
        \end{bmatrix}, \quad 
    \vz^{\sf max} :=  \begin{bmatrix}
            \vp^{\sf max}\\
            \vq^{\sf max}_\setP\\
            \vg^{\sf max}\\
            \vb^{\sf max}
        \end{bmatrix}.
\end{equation}
Let $\vy :=  [\vv^\T, \vell^\T]^\T \in \R^{n+m}$ be the approximated solutions, where $\vy := \vy^{\bullet} + \mK(\vz - \vz^{\bullet})$. The lower and upper bounds of the output are
\begin{equation}
    \vy^{\sf min} := \begin{bmatrix}
            \vv^{\sf min}\\
            \vell^{\sf min}
        \end{bmatrix}, \quad 
    \vy^{\sf max} :=  \begin{bmatrix}
            \vv^{\sf max}\\
            \vell^{\sf max}
        \end{bmatrix}.
\end{equation}

Finally, set~$\vgamma \in [\gamma^{\sf min},\gamma^{\sf max}]^m$, where~$\gamma^{\sf min},\gamma^{\sf max} \in [0,1]$, with~$\gamma^{\sf min} < \gamma^{\sf max}$. We can then form an approximated voltage regulation program where the goal is to minimize active power curtailment with controllable reactive power injections and network admittances. This program can be written as
\begin{subequations}
\label{eq:linear-program-vreg}
    \begin{align}
        \max_{\vz,\vgamma} \quad &\vone^\T \vp\\
        \st \quad  &\vy^{\sf min} \leqslant \vy^{\bullet} + \mK( \vz - \vz^{\bullet}) \leqslant \vy^{\sf max}\\
        &\vz^{\sf min} \leqslant \vz \leqslant \vz^{\sf max}\\
        &\gamma^{\sf min} \vone_m \leqslant \vgamma \leqslant \gamma^{\sf max} \vone_m\\
        &\vg = \diag(\vgamma) \vg^{\bullet}\\
        &\vb = \diag(\vgamma) \vb^{\bullet},
        % &\vvbar \leqslant \vv^{\bullet} + \mK^p \vp + \mK^q \vq + \mK \vw \leqslant \vvubar 
    \end{align}
\end{subequations}
which is a linear program. Section~\ref{sec:voltage_regulation_simulation} shows how continuous admittance control can be used for voltage regulation with the formulation in \eqref{eq:linear-program-vreg}. We also demonstrate its use for increasing the hosting capacity of a sample distribution network.

\subsubsection*{Discrete topology control for congestion management}
\label{sec:discrete_top_ctrl}
% \input{grad_top_ctrl}
%%%%%%%%%%%%%%%%%%%%%
%
% Cleaned up relaxed formulation/first pass at algorithm.
%%%%%%%%%%%%%%%%%%%%%%%%%%%%%%%%%%%
\newcommand{\vztilde}{\vct{\tilde{z}}}
Admittance sensitivities also enable fast discrete line switching algorithms, where we optimize over a vector of binary \textit{switching variables} $\vz \in \c{0,1}^m$, where $z_{ij} =1$ (resp. $z_{ij}=0$) indicates that the line $(i,j)\in\mathcal{E}$
is closed (resp. switched open). In this section, we illustrate an efficient way to optimize topologies using the framework of the present paper (Alg. \ref{alg:switchfast}), which we present below.

%%%%%%%%%%%%%%%%%%%%%%%%%%%%%%%%%%%%%%%%%%%%%%%%%%%%%%%%
% \subsubsection{Discrete problem}
%%%%%%%%%%%%%%%%%%%%%%%%%%%%%%%%%%%%%%%%
In many applications involving congestion management, the objective of a grid planning or topology control problem can be written as a parameterized function of the line admittances, or modifications thereof, i.e., switching statuses. As a simple example, consider the problem of minimizing the sum of the flow-to-capacity ratios of each line~$c_{ij}(\vx;\vz) := |u_{ij}\p{\vx(\vz)}|^{2}/(u_{ij}^{\sf \max})^{2}$, where~$\vx(\vz)$ is a power flow solution which depends on~$\vz$, see~\cite{babaeinejadsarookolaee_congestion_2023} for more details on this metric. Seeking to minimize this congestion metric yields a program of the form
\begin{subequations}\label{eq:discrete-topology}
\begin{align}
    \min_{\vz \in \c{0,1}^m,\, \vx\in\C^n}\;
        & \sum_{(i,j)\in\mathcal{E}} z_{ij}\,c_{ij}(\vx;
        \vz) =: C(\vx;\vz)\label{eq:discrete:obj}\\
    \st \quad
        &\vz \in \setZ, \label{eq:cons:K_sw}\\
        &\mathsf{power\ flow\ equations}, \label{eq:cons:acpf}\\
        &\mathsf{engineering\ constraints,}  \label{eq:eng_constr}
\end{align}
\end{subequations}
where constraint \eqref{eq:cons:K_sw} represents switching constraints, which are not included in~\eqref{eq:eng_constr}. For example, one simple choice is the set~$\setZ=\c{\vz^\T\vone \;\le\; K}$, which represents a budget of at most~$K$ switches being switched closed. The program~\eqref{eq:discrete-topology} is a challenging mixed-integer non-linear program due to the binary variables and the AC equations embedded in $c_{ij}(\cdot)$ through the constraint \eqref{eq:cons:acpf}. 

%%%%%%%%%%%%%%%%%%%%%%%%%%%%%%%%%%%%%%%%%%%%%%%%%%%%%%%%
\subsubsection*{Continuous relaxation and rounding}
\label{sec:cont-relax-linear}
%%%%%%%%%%%%%%%%%%%%%%%%%%%%%%%%%%%%%%%%%%%%%%%%%%%%%%%%
Let $\vgamma\in[0,1]^m$ be the relaxed switch vector. Using the current sensitivity results from Section~\ref{sec:diff-curr-wrt-w}, the gradient of the squared current magnitude through each line~$l$ w.r.t. the relaxed switching variables is
\begin{equation}
\label{eq:current_grad}
\mathbf{\mathfrak{g}} := \nabla_{\vgamma} \abs{u_l(\vgamma)}^2 = \frac{\partial}{\partial \vgamma} \abs{u_l}^2 = 2\cdot\abs{u_l(\vgamma)} \vk_l,
\end{equation}
where~$(\vk_l)_e =\partial \abs{u_l}/\partial\gamma_e$ are the sensitivities of the current magnitude through line~$l$ w.r.t. the admittance parameter~$\gamma_e$. 

In particular, in the simple case where~$\setZ = \c{\vz : \vz^\T\vone \leq K}$, we can use the gradients~\eqref{eq:current_grad} to approximate the objective of the program~\eqref{eq:discrete-topology}, which results in the linear program
\[
\label{eq:knapsack}
\begin{aligned}
  \min_{\vzero\le\vgamma\le\vone}\quad &\vc^{\top}\vgamma \\
  \text{s.t.}\quad & \mathbf 1^{\top}\vgamma\le K,
\end{aligned}
% \tag{$\mathcal P_{\text{relax}}$}
\]
where~$\vc:=[c_1,\dots,c_m]^{\top}$ contains the coefficients of a first-order Taylor series approximation of the objective function in~\eqref{eq:discrete:obj} as a function of~$\vgamma$. 

In some cases, it may be tractable to solve the relaxed problem deterministically and ensure integrality; one such example is the aforementioned case of the knapsack constraint. This approximated problem can be solved to optimality with~Alg.~\ref{alg:switchfast}, followed by a tractable AC feasibility step which, when satisfied, ensures the algorithm returns a feasible solution to the program~\eqref{eq:discrete-topology}.
\begin{algorithm}[ht]
\DontPrintSemicolon
\KwIn{Network data, budget constraint, initial config.}
% \KwOut{Solution for \eqref{eq:knapsack}}
\KwOut{Feasible solution for~\eqref{eq:discrete-topology}}
\BlankLine
\SetKwBlock{Begin}{function}{end function}
\Begin($\textsc{QuickSwitch} {(} \setG,\setZ,\vz^\bullet {)}$)
{
    $\vx_\star \gets \mathsf{ACPF}(\vz^\bullet)$ \;

    \While{$\vx_\star \ \mathsf{infeasible}$}{
    $\boldsymbol{\mathfrak{g}} \gets \nabla_{\vgamma} C\p{\vx_\star}$ \tcp*{implicit diff.}
    $\vz_\star\leftarrow\arg\min_{\vz\in\setZ}
                \ip{\boldsymbol{\mathfrak{g}}}{\vz}$ \tcp*{lin. approx.}   % lowest-cost 
    $\vx_\star \gets \mathsf{ACPF}(\vz_\star)$\;
    }
    \Return{$\vz_\star,\vx_\star$}\;
    % \Return{$\vz_\star\gets\vztilde$}\;
}
\caption{Greedy AC knapsack switching algorithm}
\label{alg:switchfast}
\end{algorithm}

In cases where the switching constraint set~$\setZ$ remains this intractable, it is possible to interpret each optimized weight~$\gamma_e^{\star}$ as a Bernoulli parameter; in which case, we can draw a random switch vector~$\tilde{\vz}$ with~$z_{ij}\sim\mathsf{Ber}(\gamma^{\star}_{ij})$, evaluate constraint violations, and update lines with gradient information until feasibility is restored. This is essentially randomized rounding with a feasibility step; Alg.~\ref{alg:switchfast} can be augmented accordingly.

% ============================================================
% NUMERICAL EXPERIMENTS
% ============================================================
\section{Numerical Experiments}
\label{sec:experiments}

Our first numerical illustration will use a simple 5-bus network with two possible radial topologies, shown in the diagram provided in Fig. \ref{fig:simple-network}. The test case is based on a subnetwork of the Baran and Wu 33-bus test case \cite{baran-reconfiguration-1989}. %This network is publicly available.\footnote{\url{{https://matpower.org/docs/ref/matpower6.0/case33bw.html}}} 
The two possible radial topologies of this network are used to demonstrate that the proposed differentiation technique can be used to linearize around network admittance changes. 

% \begin{remark}
%     In addition to the closed-form expressions in \eqref{eq:p-to-w-derivatives} and \eqref{eq:q-to-w-derivatives}, the matrices $\partial \mF/\partial \vg$ and $\partial \mF / \partial \vb$ can be quickly computed in practice via automatic differentiation. The power flow Jacobian $\mJ$ can similarly be computed with well-known closed form expressions, automatic differentiation, or publicly available power flow analysis libraries. Finally, the voltage-topology sensitivities \eqref{eq:implicit-function-theorem-sol} can be efficiently evaluated with a sparse linear solver.
% \end{remark}

\begin{figure}[!h]
    \centering
    \begin{adjustbox}{scale=0.75}
        \begin{tikzpicture}[thick]
    % \draw[step=0.5, lightgray, thin] (0.0,0.0) grid (8.0, 8.0);
    
    % SADR: generator at bus 1.
    \node[vsourcesinshape, name=G1, rotate=90] at (0.0, 2.75){};
    \draw (G1.south) -- (1.25, 2.75);
    
    % SADR: bus 1.
    \draw[line width=3pt] (1.25, 3.50) node[above]{$1$} --(1.25, 2.0);
    
    % SADR: GSU transformer.
    \draw (1.25, 2.75) to[oosourcetrans] (3.00, 2.75);

    % SADR: bus 2.
    \draw[line width=3pt] (3.00, 3.50) node[above]{$2$} -- (3.00, 2.0);

    % SADR: upper and lower branches.
    \draw (3.00, 2.75) to[short] (3.50, 2.75);
    \draw (3.50, 2.75) -- ++(0.00, 1.00) to[generic] (5.65, 3.75);
    \draw (3.50, 2.75) -- ++(0.00, -1.00) to[generic] (5.65, 1.75);

    % SADR: bus 3.
    \draw[line width=3pt] (5.65, 4.50) node[above]{$3$} -- (5.65, 3.0);
    % SADR: load at bus 3.
    \draw[-{Triangle[]}] (5.65, 3.25) -| ++(0.50, -0.75);

    % SADR: bus 5.
    \draw[line width=3pt] (5.65, 1.00) node[below]{$5$} -- (5.65, 2.50);
    % SADR: load at bus 5.
    \draw[-{Triangle[]}] (5.65, 1.25) -| ++(0.50, -0.75);

    % SADR: breaker between buses 5 and 4.
    \draw(5.65, 3.75) -- (6.5, 3.75) to[short, o-o] (7.25, 3.75) -| ++(0.50, -1.00);

    % SADR: breaker between buses 5 and 4.
    \draw(5.65, 1.75) -- (6.5, 1.75) to[nos, o-o] (7.25, 1.75) -| ++(0.50, 1.00);

    % SADR: bus 4.
    \draw (7.75, 2.75) -- ++(0.50, 0.00);
    \draw[line width=3pt] (8.25, 3.50) node[above]{$4$} -- (8.25, 2.0);
    % SADR: load at bus 4.
    \draw[-{Triangle[]}] (8.25, 2.75) -| ++(0.50, -0.75);
\end{tikzpicture}
    \end{adjustbox}
    \caption{The simple 5-bus test case  shown in its final configuration ($\gamma=1$).}
    \label{fig:simple-network}
\end{figure}
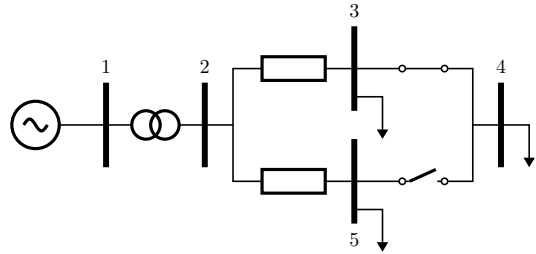

\subsection{Power flow solution prediction as admittance changes}
\label{sec:experiments:predict_sol}

In Fig. \ref{fig:linearization-illustration}, we show the linearization of the simple 5-bus switching network. The top plot in particular illustrates the significant performance gains from linearizing about the midpoint of the interval of possible admittance parameters.
\begin{figure}
    \centering
    \includegraphics[width=0.95\linewidth,keepaspectratio]{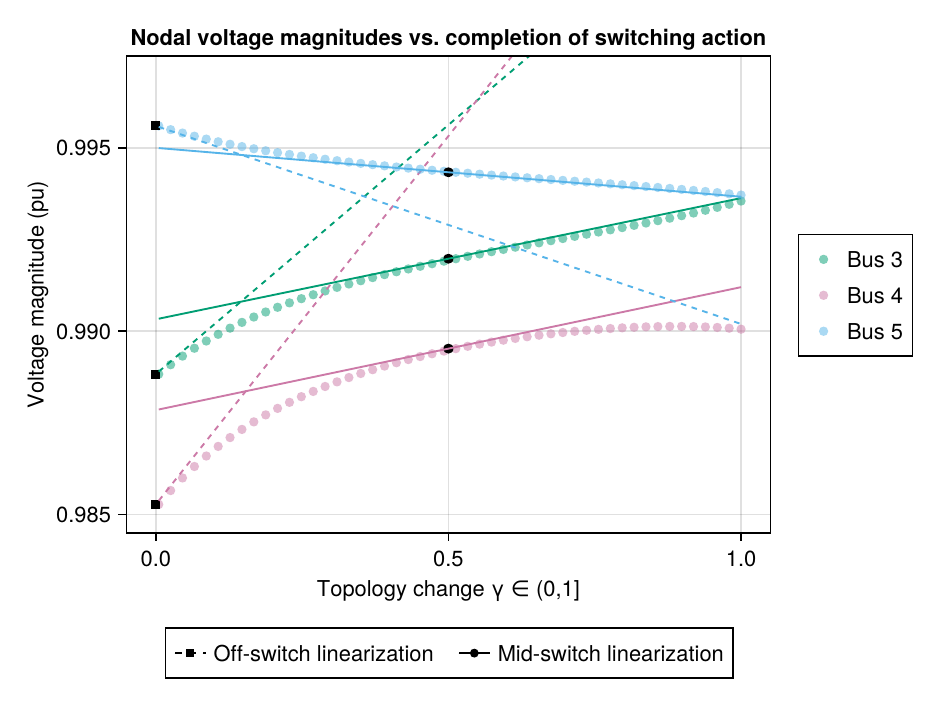}
    \includegraphics[width=0.95\linewidth,keepaspectratio]{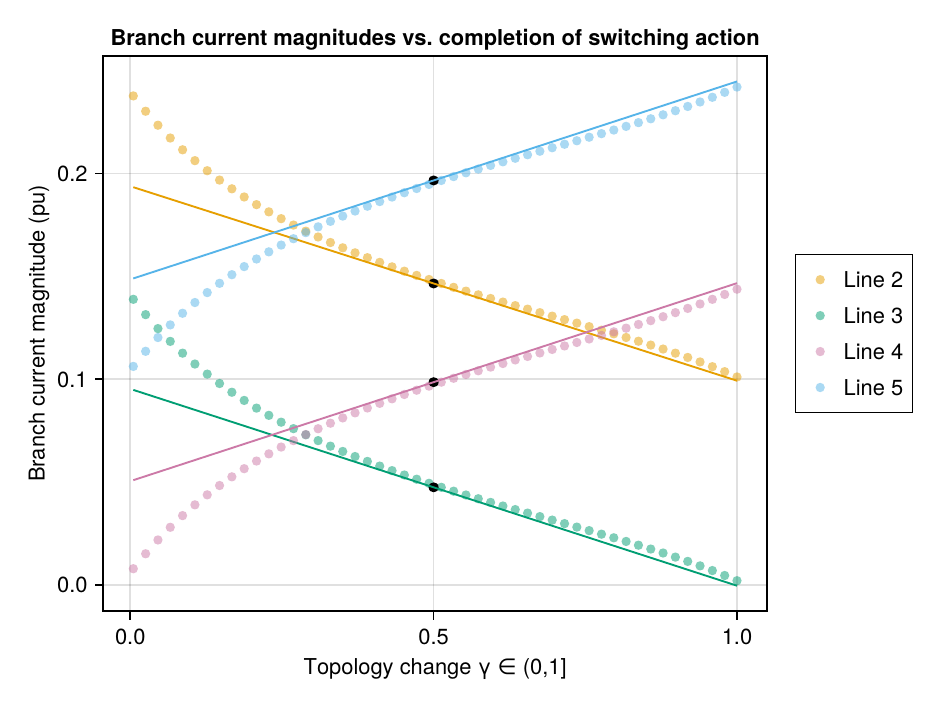}
    \vspace{-10pt}
    \caption{Illustrative linearization of nodal voltage magnitudes (top) and branch current magnitudes (bottom) as a function of the network topology change for the simple test case shown in Fig. \ref{fig:simple-network}. }
    \label{fig:linearization-illustration}
\end{figure}

We expand the result to the full 33-bus test case proposed by Baran and Wu and show the results in Fig.~\ref{fig:full-bw}. In particular, this plot shows the performance of the linearization for predicting the AC power flow solution for the power set of the switchable lines, which contains all possible topology configurations. The top figure shows the performance of repeated linearization about the average of the base case and the topology scenario admittance (gold), compared to ``doing nothing", i.e., using the base case AC power flow solution (blue). The bottom figure shows a comparison of the results in the top figure with two other methods: averaging all $2^{\abs{\setE}}$ topology scenarios (green) and averaging between the two admittance parameters with all switches on and all switches off (purple).
\begin{figure}
    \centering
    \includegraphics[width=0.95\linewidth,keepaspectratio]{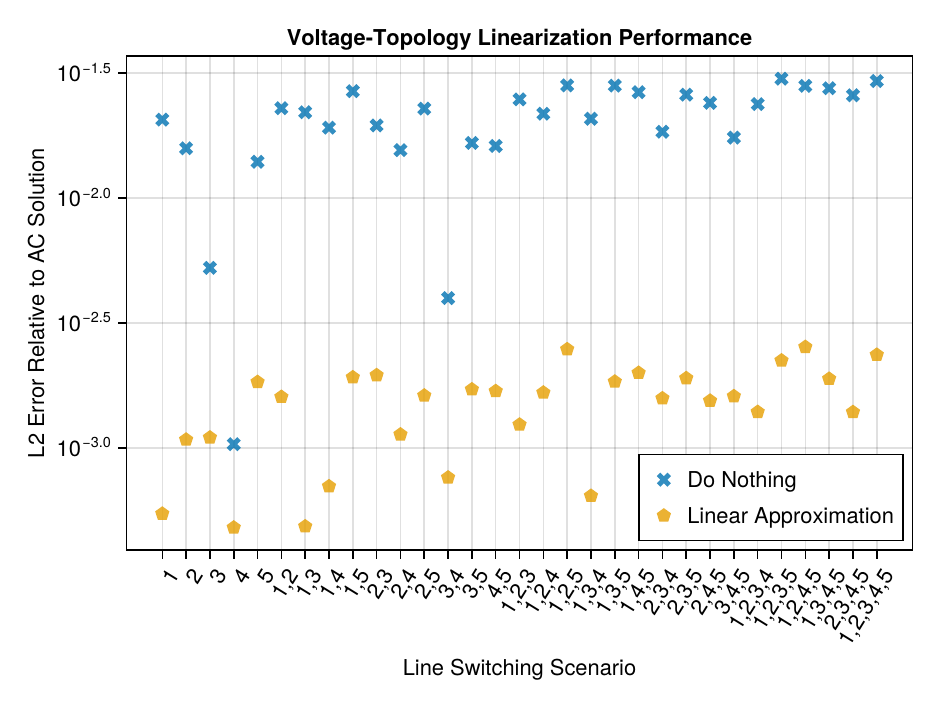}
    \includegraphics[width=0.95\linewidth,keepaspectratio]{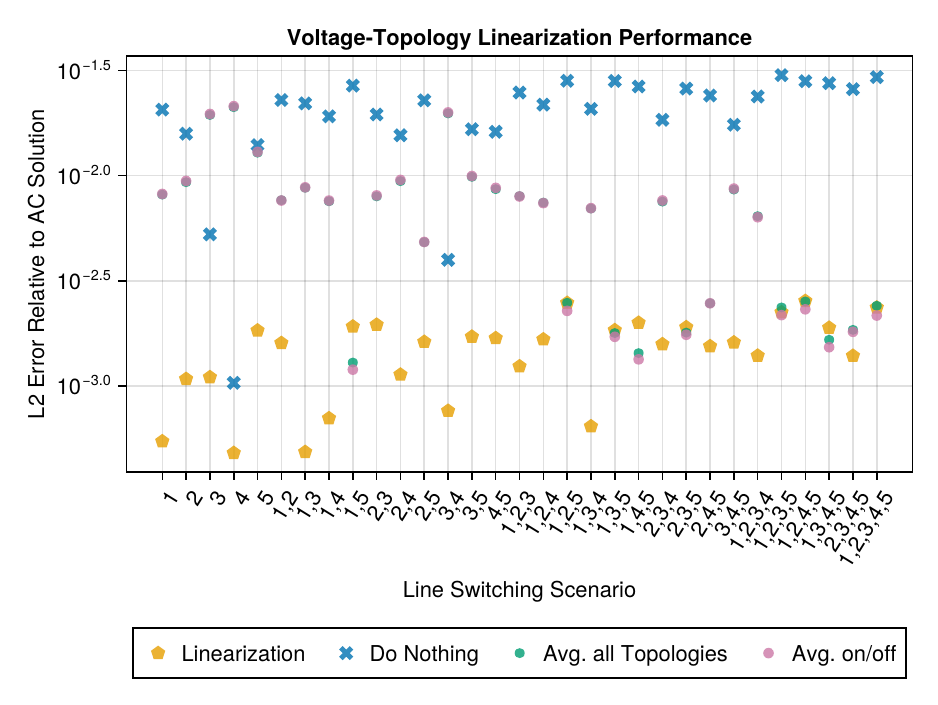}
    \caption{Relative error of approximating the AC power flow solution produced by the power set of switching configurations for the \texttt{case33bw} radial network.}
    \label{fig:full-bw}
\end{figure}

\subsection{Continuous Admittance Control}
\subsubsection*{Voltage regulation application}
\label{sec:voltage_regulation_simulation}
We next demonstrate a voltage regulation problem that takes advantage of the proposed method, which we apply to the CIGRE low-voltage distribution network model~\cite{CIGREREF}. The results are shown in Fig. \ref{fig:vreg-results}.
\begin{figure}
    \centering
    \includegraphics[width=0.95\linewidth,keepaspectratio]{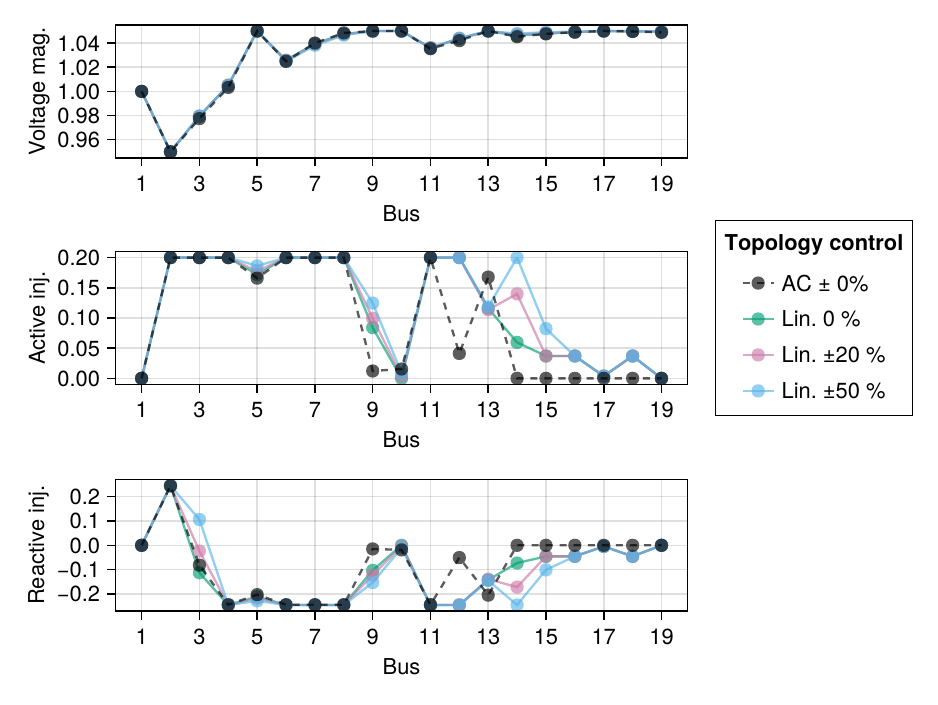}
    \caption{Optimal solutions of the voltage regulation problem for the CIGRE low voltage test case \cite{CIGREREF}. The plot compares the minimum curtailment solution using the full AC power flow equations and the linearized formulation \eqref{eq:linear-program-vreg} with admittance bounds ranging between $\pm 0 \% $ and $\pm 50\%$.}
    \label{fig:vreg-results}
\end{figure}

The performance of the method, in terms of relative error to the AC solution and the maximum active power injection for the \texttt{case33bw} test network, is shown in Fig. \ref{fig:case33-ac-err-max-p}.
\begin{figure}
    \centering
    \includegraphics[width=0.95\linewidth,keepaspectratio]{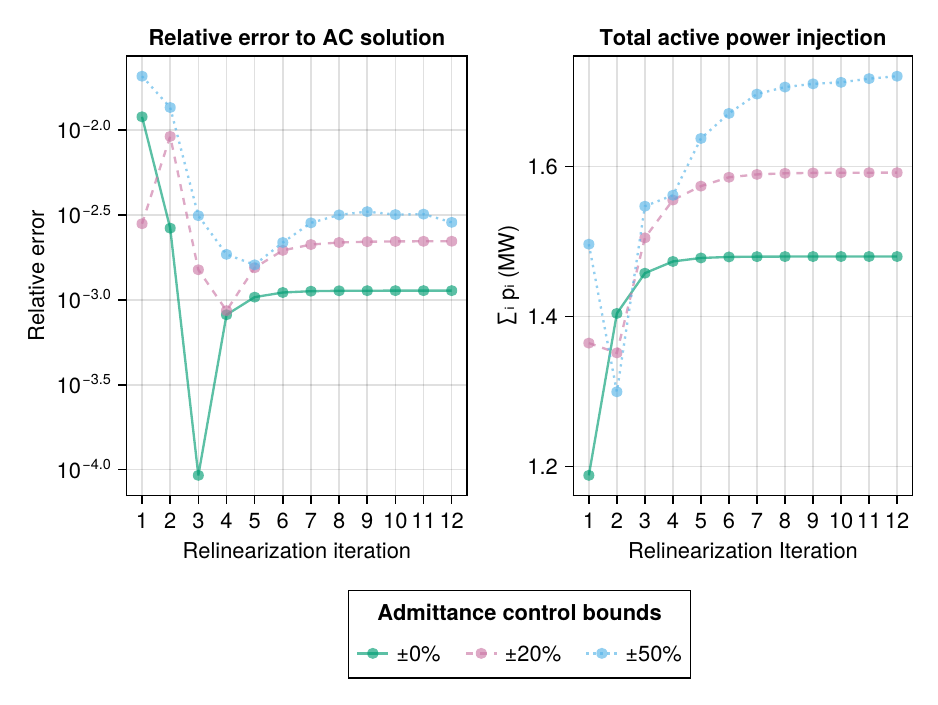}
    \caption{Relative error to AC solution (left) and maximum active power injections (right) yielded by solutions of \eqref{eq:linear-program-vreg} vs. number of successive linearizations about the midpoint. These plots only consider voltage constraints, and neglect the current constraints.}
    \label{fig:case33-ac-err-max-p}
\end{figure}

Table \ref{table:linearized-vreg-performance} shows the performance of the voltage regulation problem \eqref{eq:linear-program-vreg} across different ranges of~$\gamma$, and with iterative re-linearization for multiple test cases. The predicted network parameters, active power injections, and reactive power injections are collected after solving the problem, and are used to augment the network model, which is then used to solve a standard AC power flow. The nodal voltage magnitudes output by the AC power flow solution and predicted by the linearized program~\eqref{eq:linear-program-vreg} are then compared in the two leftmost columns of the table; this includes the relative percentage error in the sense of the~$\ell_2$ norm and the absolute error in the sense of the~$\ell_\infty$ norm. Furthermore, the penultimate rightmost column displays the maximum active power injections determined by the program, with the far rightmost column reporting the percentage increase of this quantity relative to the baseline case where~$\gamma$ is invariant. 
\begin{table}[t!]
    \caption{Relative $\ell_2$ error and absolute $\ell_\infty$ error with AC solution, active injection capacity, and increase in active injection capacity with continuous admittance control}
    \centering
    \begin{tabular}{clllll}
    \toprule
    \multicolumn{1}{l}{\textbf{Case}} &
    $\gamma$ &
    \multicolumn{1}{l}{$\ell_2$ (\%)} &
    \multicolumn{1}{l}{$\ell_\infty$ (pu)} &
    \multicolumn{1}{l}{$\vone^\T \vp^*$ (MW)} &
    \multicolumn{1}{l}{$\uparrow$ (\%)} 
    \\ 
    \midrule
    \multirow{3}{*}{cigreLV} & $\pm 0$  & 0.01 & 2.26 $\times 10^{-4}$ & 2.42 & -     \\
                            & $\pm 20$ & 0.28 & 3.28 $\times 10^{-3}$ & 2.77 & 14.65 \\
                            & $\pm 40$ & 1.48 & 0.016    & 3.11 & 28.50 \\ \hdashline
    \multirow{3}{*}{33bw}    & $\pm 0$   & 0.11 & 1.64 $\times 10^{-3}$ & 1.48 & -       \\
                            & $\pm 20$ & 0.22 & 3.08 $\times 10^{-3}$ & 1.59 & 7.50    \\
                            & $\pm 40$ & 0.32 & 5.00 $\times 10^{-3}$ & 1.71 & 15.28   \\
                            \hdashline
    \multirow{3}{*}{69bw}   & $\pm 0$   & 0.45 & 0.016    & 8.26 & -       \\
                            & $\pm 20$ & 0.40 & 0.014    & 8.68 & 4.98    \\
                            & $\pm 40$ & 0.39 & 0.013    & 9.03 & 9.38 \\ 
    \bottomrule
    \end{tabular}
\label{table:linearized-vreg-performance}
\end{table}

\subsection{Implications and Analysis of Results}
\label{sec:implications}
The contributions of our work have several engineering implications, from both practical and theoretical perspectives. One of the most immediate theoretical consequences of the technique outlined in Section \ref{sec:application-linearization} is centrally showcased in the numerical results in Section \ref{sec:experiments:predict_sol}, showing the \textit{prediction of power flow solutions} as a function of network admittance parameters. 

Conventionally, determining power flow solutions across a range of network admittances would require repeatedly executing a power flow solver across $k$ intervals, as shown in the dotted lines in Fig. \ref{fig:linearization-illustration}. If a conventional Newton-Raphson algorithm was executed for $m$ iterations, sweeping across the admittance space would require $km$ Jacobian factorizations. In contrast, applying the proposed linearization technique allows us to \textit{estimate the outcome of this iterative procedure} with a precomputed Jacobian with respect to admittance. This prediction would amount to $k$ matrix multiplications.

%\subsection{Limitations and unanswered questions}
%\label{sec:limitations}
%The presented results have produced new insights into approximating the power flow equations\textemdash with several immediate and practically relevant applications in efficiently predicting power flow solutions across known  admittance changes, and continuous admittance controllers. Notwithstanding these results, there are clear shortcomings of the method in application to \emph{discrete} topology control. 

%\subsection{Future work}
%\label{sec:discussion-future-work}
%We identify several directions for future work, which we organize in terms  of theory-focused and application-focused directions.

%In terms of theoretical improvements, the optimization of the linearization procedure developed in Section \ref{sec:application-linearization} is an important avenue of future work proposed by the authors. When we say optimization of the linearization procedure, we specifically refer to:
% \begin{enumerate}
%     \item Optimization of the nominal point about which the linearization is taken. This is clearly motivated by the contrasting performance between linearizing about the centroid of the admittances arising from pairs of topology configuration, and linearizing about the centroid of multiple topology configuration\textemdash \textit{i.e.}, as shown in Fig. \ref{fig:linearization-illustration} and Fig. \ref{fig:full-bw}.
% \end{enumerate}

\section{Conclusion}
\label{sec:conclusion}
% In this article, we derived a method for implicitly differentiating power flow solutions with respect to network admittance parameters, and demonstrated practical applications of this method. We showed that this implicit differentiation technique is guaranteed to work under mild conditions; namely, that an AC power flow solution has been successfully computed, and that the network is not in a state of voltage collapse. The
In this work, we developed a rigorous framework to compute the sensitivities of power flow solutions with respect to network admittance parameters. Our method efficiently derives these sensitivities by differentiating the bus injection model of the power flow equations and leverages the implicit function theorem to derive admittance sensitivities. We established that these coefficients can be obtained by solving a linear system of equations, which remains uniquely solvable as long as the power flow Jacobian is invertible. The scheme was applied to compute sensitivity of complex voltages,  line currents and power flows.

We demonstrated the application of these sensitivities in obtaining a linearized formulation for continuous admittance control, enabling computationally efficient network control. We evaluated the proposed method on multiple benchmark networks demonstrating the tractability of admittance control in power networks. The proposed framework paves the way for tractable linear formulation for adaptive network reconfiguration, optimal line switching, and continuous admittance control operational flexibility.

% SADR: IEEE style.
\bibliographystyle{references/IEEEtran}
% SADR: uncomment after adding *.bib file.
\bibliography{references/refs}

% \appendix[Further details on the differentiation of power injections with respect to admittance parameters]
\appendix

\subsection{Further details on the differentiation of power injections with respect to admittance parameters}
\label{apdx:diff-power-to-params}
In the expression \eqref{eq:implicit-function-theorem-sol}, the Jacobian with respect to the admittance parameters can be expanded into a~$2 \abs{\setN} \times 2 \abs{\setE}$ block matrix of the form
\begin{equation}
\begin{bmatrix}
    \frac{\partial \vkappa}{\partial \vg} & \frac{\partial \vkappa}{\partial \vb}
\end{bmatrix} :=
    \begin{bmatrix}
        \frac{\partial \vp}{\partial \vg} & \frac{\partial \vp}{\partial \vb} \\
        \frac{\partial \vq_\setP}{\partial \vg} & \frac{\partial \vq_\setP}{\partial \vb}
    \end{bmatrix}.
\end{equation}
Each block is a~$\abs{\setN} \times \abs{\setE}$ Jacobian matrix taken with respect to the conductance or susceptance parameters $\vg$ and $\vb$, respectively. The derivatives that comprise the entries of these matrices
can be analytically computed from~\eqref{eq:special-pf-eq-representation}. In what follows, $\mathds{1}\{\cdot\}$ denotes the indicator function. When no phase-shifting transformers are present, the entries corresponding to the active power injections are given as
%===================== Daniel V1
% \begin{subequations}
%     \label{eq:p-to-w-derivatives}
% \begin{align}
%     \frac{\partial p_i}{\partial g_{k j}} &= v^{2}_{k} - \mathds{1}\left\{j \neq k\right\}  v_k v_j  \cos\left(\delta_{kj}\right), \\
%     \frac{\partial p_i}{\partial b_{k j}} &= - \mathds{1}\left\{j \neq k\right\} v_k v_j \sin\left(\delta_{kj}\right),
% \end{align}
% \end{subequations}
%====================== Daniel V2
\begin{subequations}
    \label{eq:p-to-w-derivatives}
\begin{align}
    \frac{\partial p_i}{\partial g_{i j}} &= \frac{\partial p_i}{\partial g_{ji}} = v^{2}_{i} - \mathds{1}\left\{j \neq i\right\}  v_i v_j  \cos\left(\delta_{ij}\right), \\
    \frac{\partial p_i}{\partial g_{k j}} &= 0, \quad k,j \neq i, \\
    \frac{\partial p_i}{\partial b_{i j}} &= \frac{\partial p_i}{\partial b_{ji}} = - \mathds{1}\left\{j \neq i\right\} v_i v_j \sin\left(\delta_{ij}\right), \\
    \frac{\partial p_i}{\partial b_{k j}} &=  0, \quad k,j \neq i,
\end{align}
\end{subequations}
for all $i \in \setN$, $(k,j) \in \setE$. Analogously, in the absence of phase-shifting transformers, the entries corresponding to the reactive power injections are
%===================== Daniel V1
% \begin{subequations}
% \label{eq:q-to-w-derivatives}
%     \begin{align}
%         \frac{\partial q_i}{\partial g_{kj}} &= - \mathds{1}\left\{ j \neq k\right\}v_k v_j \sin(\delta_{kj}),\\
%         \frac{\partial q_i}{\partial b_{kj}} &= - v_i^2 + \mathds{1}\left\{ j \neq k \right\} v_i v_j \cos(\delta_{ij}),
%     \end{align}
% \end{subequations}
%====================== Daniel V2
\begin{subequations}
\label{eq:q-to-w-derivatives}
    \begin{align}
        \frac{\partial q_i}{\partial g_{ij}} &= \frac{\partial q_i}{\partial g_{ji}} = - \mathds{1}\left\{ j \neq i\right\}v_i v_j \sin(\delta_{ij}), \\
        \frac{\partial q_i}{\partial g_{kj}} &= 0, \quad k,j \neq i, \\
        \frac{\partial q_i}{\partial b_{ij}} &= \frac{\partial q_i}{\partial b_{ji}} - v_i^2 + \mathds{1}\left\{ j \neq i \right\} v_i v_j \cos(\delta_{ij}), \\
        \frac{\partial q_i}{\partial g_{kj}} &= 0, \quad k,j \neq i,
    \end{align}
\end{subequations}
for all $i \in \setP$, $(k,j) \in \setE$.

\end{document}